\documentclass[conference]{IEEEtran}

\usepackage{amssymb,amsmath,amsfonts,amsthm}
\usepackage{algorithm,algorithmic}
\usepackage{boxedminipage}
\usepackage{cite}
\usepackage[dvips]{graphicx}
\usepackage{epsfig}
\usepackage{float}
\graphicspath{{figures/}}
\usepackage{graphicx}
\usepackage{graphics}
\usepackage{multirow}
\usepackage{threeparttable} % package for footnotes in tables
\usepackage[usenames,dvipsnames]{color}

\newtheorem{theorem}{Theorem}
\newtheorem{lemma}{Lemma}

\newtheorem{remark}{Remark}

\newtheorem{problem}{Problem}

\newcommand{\sr}{\stackrel}

\newcommand{\rar}{\rightarrow}

\newcommand{\tri}{\sr{\triangle}{=}}

\newcommand{\be}{\begin{equation}}
\newcommand{\ee}{\end{equation}}
\newcommand{\bea}{\begin{eqnarray}}
\newcommand{\eea}{\end{eqnarray}}
\newcommand{\bes}{\begin{eqnarray*}}
\newcommand{\ees}{\end{eqnarray*}}

\newcommand{\bi}{\begin{itemize}}
\newcommand{\ei}{\end{itemize}}
\newcommand{\ben}{\begin{enumerate}}
\newcommand{\een}{\end{enumerate}}

\newcommand{\bp}{\begin{problem}}
\newcommand{\ep}{\end{problem}}
\newcommand{\hso}{\hspace{.1in}}
\newcommand{\hst}{\hspace{.2in}}

\newcommand{\noi}{\noindent}

\newcommand{\bc}{\begin{center}}
\newcommand{\ec}{\end{center}}

\floatstyle{ruled}
\newfloat{algorithm}{H}{algo}
\floatname{algorithm}{\footnotesize Algorithm}

\begin{document}
\title{Lossless Coding with Generalized Criteria  }

\author{\IEEEauthorblockN{Charalambos D. Charalambous}
\IEEEauthorblockA{Department of Electrical and\\Computer Engineering\\
University of Cyprus\\
Email: chadcha@ucy.ac.cy}
\and
\IEEEauthorblockN{Themistoklis Charalambous}
\IEEEauthorblockA{Department of Electrical and\\Computer Engineering\\
University of Cyprus\\
Email: themis@ucy.ac.cy}
\and
\IEEEauthorblockN{Farzad Rezaei}
\IEEEauthorblockA{School of Information Technology \\ and Engineering\\
University of Ottawa\\
Email: frezaei@site.uottawa.ca}
}

\maketitle

% =====================================================
%
%
% ABSTRACT
%
%
% =====================================================
\begin{abstract}
This paper presents prefix codes which minimize various criteria constructed as a convex combination of maximum codeword length and average codeword length or maximum redundancy and average redundancy, including a convex combination of the average of an exponential function of the codeword length and the average redundancy. This framework encompasses as a special case several criteria previously investigated in the literature, while relations to universal coding is discussed. The coding algorithm derived is parametric resulting in re-adjusting the initial source probabilities via a weighted probability vector according to a merging rule. The level of desirable merging has implication in applications where the maximum codeword length is bounded.
\end{abstract}

\IEEEpeerreviewmaketitle

% =====================================================
%
%
% INTRODUCTION
%
%
% =====================================================
\section{Introduction}
Lossless fixed to variable length source codes are usually examined under known source  probability distributions, and unknown source probability distributions. For known source probability distributions there is an extensive literature which aims at minimizing various pay-offs such as the average  codeword length \cite{2006:Cover}, the average redundancy of the codeword length \cite{2004:DrmotaSzpankowski,2008a:Baer}, the average of an exponential function of the codeword length    \cite{1965:campbell_coding,1981:humblet_generalization,2008b:Baer}, the average of an exponential function of the redundancy of the codeword length \cite{2006a:Baer,2008a:Baer, 2008b:Baer}.  On the other hand, universal coding and universal modeling, and the so-called Minimum Description Length (MDL) principle  are often examined via minimax techniques, when the source probability distribution is unknown, but belongs to a pre-specified  class of source distributions  \cite{1973:davisson,1980:davisson_Leon-Garcia, 2004:DrmotaSzpankowski,2007:FarzadCharalambous ,2009:Gawrychowski_Gagie}. %With respect to the above pay-offs Shannon codes find  sub-optimal code lengths by treating them as real numbers, while Huffman codes  find the optimal  code lengths by treating them as integers. Although Shannon codes are over the years investigated for a variety of pay-offs, optimal Huffman codes are available only for a limited  number of pay-off functions.

This paper is concerned with lossless coding problems, in which the pay-offs are the following.  1) A convex combination of the maximum codeword length and the average codeword length, or a convex combination of the maximum pointwise redundancy and the average pointwise redundancy of the codeword length, and 2) a convex combination of the average of an  exponential function of the codeword length and the average codeword length, or a convex combination of the average of an exponential function of the pointwise redundancy and the average redundancy of the  codeword length. 

These are multiobjective pay-offs whose solution bridges together an anthology of source coding problems with different pay-offs including some of ones investigated  in the above mentioned references. Moreover, for 1) there is  parameter $\alpha \in [0,1]$ which weights the maximum codeword length (resp. maximum pointwise redundancy of the codeword) while $(1-\alpha)$ weights the average codeword length (resp. average redundancy of the codeword), and as this parameter moves away from $\alpha=0$ the maximum length of the code is reduced resulting in a more balanced code tree. A similar conclusion holds  for 2) as well.

%{\bf discuss shortly the results}\\

% =====================================================
%
%
% PROBLEM FORMULATION
%
%
% =====================================================
\subsection{Objectives  and Related Problems}\label{sec:statement}

\noi Consider a source with alphabet  ${\cal X } \tri \{x_1, x_2, \ldots, x_{ | {\cal X }| }  \}$  of cardinality   $|{\cal X}|$, generating symbols according to  the probability distribution ${\bf p} \tri \{p(x): x \in {\cal X}  \} \equiv \big(p(x_1), p(x_2), \ldots, p(x_{|{\cal X}|})\big)  $. Source symbols are encoded into $D-$ary codewords. A code ${\cal C} \tri \{c(x): x \in {\cal X}\}$ for symbols in ${\cal X}$ with image alphabet ${\cal D} \tri \{0, 1, 2, \ldots, D-1 \}$ is an injective map    $c: {\cal X} \rar {\cal D}^*$, where ${\cal D}^*$ is the set of finite sequences drawn from ${\cal D}$.  For $x \in {\cal X}$  each codeword $c(x) \in {\cal D}^*,  c(x) \in {\cal C}$ is identified with a codeword length $l(x) \in {\mathbb Z}_+$, where ${\mathbb Z}_+$ is the set of non-negative integers. Thus, a  code ${\cal C}$ for source symbols from the alphabet ${\cal X}$ is associated with the length function of the code $ l : {\cal X} \rar {\mathbb Z}_+$, and a code defines a codeword length vector ${\bf l} \tri \{ l(x): x \in {\cal X}\} \equiv \big(l(x_1), l(x_2), \ldots, l(x_{|{\cal X}|})\big) \in {\mathbb Z}_+^{|{\cal X}|}$.  Since a function $l : {\cal X} \rar {\mathbb Z}_+$ is the length function of some prefix code if and only if it satisfies the Kraft inequality \cite{2006:Cover}, then the admissible set of codeword length vectors is defined by
\bes
{\cal L}( {\mathbb Z}_+^{|{\cal X}|}  )  \tri \Big\{ {\bf l} \in {\mathbb Z}_+^{|{\cal X}|} : \sum_{x \in {\cal X}} D^{-l(x)} \leq 1 \Big\}.
\ees
On the other hand, if the integer constraint is relaxed by admitting real-valued length vectors ${\bf l} \in {\mathbb R}^{|{\cal X}|}$ which satisfy the Kraft inequality, such as Shannon codes or arithmetic codes, then ${\cal L}( {\mathbb Z}_+^{|{\cal X}|})$ is replaced by 
\bes
{\cal L}( {\mathbb R}_+^{|{\cal X}|})  \tri \Big\{ {\bf l} \in {\mathbb R}_+^{|{\cal X}|} : \sum_{x \in {\cal X}} D^{-l(x)} \leq 1 \Big\}.
\ees
%Such codes are useful in obtaining approximate solutions which are less computational intensive \cite{cover}.
Without loss of generality is it is assumed that the set of probability distributions is defined by

\bea
{\mathbb P}({\cal X}) \tri \Big\{ {\bf p} =\big(p(x_1), \ldots, p(x_{|{\cal X}|})\Big) \in {\mathbb R}_+^{|{\cal X}|} : p(x_{|{\cal X}|})>0, \nonumber \\ 
p(x_i) \leq p(x_j), \forall i>j,  (x_i,x_j) \in {\cal X},  \sum_{x \in {\cal X}} p(x) =1 \Big\}. \nonumber
\eea

\noi Moreover, $\log (\cdot) \tri \log_D (\cdot)$ and ${\mathbb H}({\bf p})$ denotes the entropy of the probability distribution ${\mathbf p}$. The two main problems investigated  are the following.

% ----------------------------------------------------------------
% Problem 1
% ----------------------------------------------------------------
\begin{problem}\label{problem1}
Given a known source probability vector ${\bf p} \in   {\mathbb P}({\cal X})$ define the one parameter pay-off
\vspace{-0.2cm}
\bea
{\mathbb L}_{\alpha}^{MO}({\bf l}, {\bf p}) \tri \Big\{ \alpha \max_{ x\in {\cal X}} l(x) + (1-\alpha) \sum_{ x \in {\cal X}} l(x) p(x)\Big\}, \label{b30}
\eea
and a slightly general version representing redundancy
 \begin{align}
  {\mathbb L} {\mathbb R}_{\alpha}^{MO}({\bf l}+\log {\bf p}, {\bf p}) & \tri
 \alpha \max_{ x \in {\cal X}} \Big( l(x) + \log p(x)\Big) \nonumber \\
& + (1-\alpha) \Big( \sum_{ x \in {\cal X}} l(x) p(x) - {\mathbb H}({\bf p})\Big) \label{lm12}
 \end{align}
where $\alpha \in [0,1] $ is a weighting parameter.  The objective is to find a prefix code length vector ${\bf l}^* \in {\mathbb R}_+^{|{\cal X}|}$ which minimizes the pay-off ${\mathbb L}_{\alpha}^{MO}({\bf l}, {\bf p})$ or ${\mathbb L} {\mathbb R}_{\alpha}^{MO}({\bf l}+\log_D {\bf p}, {\bf p})$, for $\forall \alpha \in [0,1]$.
\end{problem}

\noi The pay-off ${\mathbb L}_{\alpha}^{MO}({\bf l}, {\bf p})$ is a convex combination of the maximum  and the average codeword length, and hence $\alpha$  weights how much emphasis is placed on the maximum and the average codeword length. The extreme cases,  $\alpha=0$ corresponds to the average codeword length, and  $\alpha=1$  corresponds to the maximum codeword length.  The pay-off ${\mathbb L} {\mathbb R}_{\alpha}^{MO}({\bf l}+\log {\bf p}, {\bf p})$ is a convex combination of the maximum pointwise redundancy and the average redundancy of the codeword length. The maximum pointwise redundancy is clearly the maximum difference between the length of the compressed symbol $l(x)$ and the self-information of that symbol $-\log p(x)$, hence this maximum redundancy is minimized over the code lengths. To the best of our knowledge neither pay-offs defined in Problem~\ref{problem1} are addressed in the literature.
Another class of problems which is also not discussed in the literature is the following. 

% ----------------------------------------------------------------
% Problem 2
% ----------------------------------------------------------------
\begin{problem}\label{problem2}
 Given a known source probability vector ${\bf p} \in   {\mathbb P}({\cal X})$ define the two parameter pay-off
\begin{align}
{\mathbb L}_{t,\alpha}^{MO}({\bf l}, {\bf p}) &  \tri \frac{\alpha}{t} \log \Big(\sum_{ x \in {\cal X}} p(x) D^{t \l(x)}\Big) \nonumber \\
 &  + (1-\alpha) \sum_{ x \in {\cal X}} l(x) p(x), \label{b30aa}
\end{align}
and a slightly general version representing redundancy
\begin{align}
{\mathbb L} {\mathbb R}_{t,\alpha}^{MO}({\bf l}+\log {\bf p}, {\bf p}) &\tri \alpha \frac{1}{t}\log \Big(\sum_{ x \in {\cal X}} p(x) D^{t \:\big( l(x) + \log p(x)\big)}\Big) \nonumber\\
& + (1-\alpha) \Big( \sum_{ x \in {\cal X}} l(x)p(x) - {\mathbb H}({\bf p})\Big)  \label{gmo1}
\end{align}
where $\alpha \in [0,1] $ is a weighting parameter and $t \in (-\infty, \infty)$. The objective is to find a prefix code length vector ${\bf l}^* \in {\mathbb R}_+^{|{\cal X}|}$ which minimizes the pay-off ${\mathbb L}_{t,\alpha}^{MO}({\bf l}, {\bf p})$ or ${\mathbb L} {\mathbb R}_{t,\alpha}^{MO}({\bf l}+\log_D {\bf p}, {\bf p})$, $\forall \alpha \in [0,1]$.
\end{problem}
\noi The two parameter pay-off ${\mathbb L}_{t, \alpha}^{MO}({\bf l}, {\bf p})$ is a convex combination of the average of an exponential function of the codeword length  and the average codeword length.  The pay-off ${\mathbb L} {\mathbb R}_{t,\alpha}^{MO}({\bf l}+\log {\bf p}, {\bf p})$ is a convex combination of the average of an exponential function of the pointwise redundancy and  the average pointwise redundancy. For $\alpha=0$ or $\alpha=1$ the resulting special cases of Problem~\ref{problem2} are found in \cite{2004:DrmotaSzpankowski,1965:campbell_coding,1981:humblet_generalization,2006a:Baer,2008a:Baer, 2008b:Baer}).
%Although, Problem~\ref{problem2} is also defined for $t \leq 0$ its solution will be discussed for $t \geq 0$.

%\subsection{Related Problems: Average Pay-Offs and Universal Coding}
%\label{ae}
\noi Hence, for $\alpha=0$ or $\alpha=1$  Problem~\ref{problem1} and Problem~\ref{problem2} are related to several problems previously investigated in the literature. The special cases ${\mathbb L}_{t,1}^{MO}({\bf l}, {\bf p}), {\mathbb L} {\mathbb R}_{t,1}^{MO}({\bf l}+\log {\bf p}, {\bf p})$ are also the dual problems of universal coding problems formulated as a minimax, in which the maximization is over  a class of probability distributions which satisfy a relative entropy constraint with respect to a given fixed nominal probability distribution (see \cite{2005:rezaei_bambos}).  

Moreover, for any $\alpha \in (0,1)$  Problem~\ref{problem1} and Problem~\ref{problem2} are multiobjective problems; clearly as $\alpha$ moves away from $\alpha=0$ more emphasis will be put on minimizing the maximum codeword length or maximum pointwise redundancy for Problem~\ref{problem1}, and the exponential function of the codeword length or pointwise redundancy for Problem~\ref{problem2}.  %specifically, Problem~\ref{problem1} penalizes the maximum codeword length  $\max_{ x \in {\cal X}} l(x)$ or its maximum codeword redundancy $\max_{ x \in {\cal X}} \big( l(x) + \log p(x)\big)$, and the penalization is weighted according to the value of $\alpha \in (0,1)$. \\ 
Relations between Problem~\ref{problem1} and Problem~\ref{problem2} and other pay-offs are established by noticing the validity of the following limits (which can be easily shown).
\bea
 \lim_{t \rar \infty} \frac{1}{t}\log_D \Big(\sum_{ x \in {\cal X}} p(x) e^{t l(x)}\Big)=\max_{ x \in {\cal X}} l(x)  \label{l} \\
%\eea
%\bea
  {\mathbb L}_{\infty,\alpha}^{MO}({\bf l}, {\bf p}) \tri  \lim_{t \rar \infty} {\mathbb L}_{t,\alpha}^{MO}({\bf l}, {\bf p})= {\mathbb L}_{\alpha}^{MO}({\bf l}, {\bf p})  \label{rl2}
\eea

\vspace{-0.4cm}

\begin{align}
    {\mathbb L} {\mathbb R}_{\infty,\alpha}^{MO} ({\bf l}+\log {\bf p}, {\bf p}) & \tri  \lim_{t \rar \infty}  {\mathbb L} {\mathbb R}_{t,\alpha}^{MO}({\bf l}+\log {\bf p}, {\bf p}) \nonumber  \\
    &  ={\mathbb L} {\mathbb R}_{\alpha}^{MO}({\bf l}+\log{\bf p}, {\bf p})   \label{lm12}
 \end{align}
Since the multiobjective pay-off ${\mathbb L}_{t,\alpha}^{MO}({\bf l}, {\bf p})$ is  in the limit, as $t \rar \infty$, equivalent to $\lim_{t \rar \infty} {\mathbb L}_{t,\alpha}^{MO}({\bf l}, {\bf p})  ={\mathbb L}_{\alpha}^{MO}({\bf l}, {\bf p}), \forall \alpha \in [0,1]$, then the codeword length vector minimizing ${\mathbb L}_{t,\alpha}^{MO}({\bf l}, {\bf p})$ is expected to converge in the limit as $t \rar \infty$, to that which minimizes ${\mathbb L}_{\alpha}^{MO}({\bf l}, {\bf p})$. A similar behavior holds for the multiobjective pay-off ${\mathbb L} {\mathbb R}_{t,\alpha}^{MO}({\bf l}+\log {\bf p}, {\bf p})$. \\

% =====================================================
%
%
% MAIN RESULTS
%
%
% =====================================================

\vspace{-0.2cm}

% =====================================================
% OPTIMAL WEIGHTS AND MERGING RULES
% =====================================================
\section{Problem~\ref{problem1}: Optimal Weights and Merging Rule}\label{subsec:weights}
\label{ow}
The objective of this section is to convert the multiobjective pay-off  of Problem~\ref{problem1}  into one which is equivalent to a single objective of the form $\sum_{x \in {\cal X}} w_\alpha(x) l(x) $,   in which $w_\alpha(x), x \in {\cal X}$ are the new weights which depend continuously on the parameter $\alpha \in [0,1]$. Subsequently, we derive certain properties of these weights associated with the optimal codeword lengths. The main issue here is to identify the combination rule of merging symbols together, and how this combination rule will change as a function of the parameter $\alpha \in [0,1]$ so that a solution exists over $[0,1]$. From these properties the Shannon codeword lengths for Problem~\ref{problem1} will be found. 

\vspace{0.1cm}

\noi Define
%\bes
$\displaystyle l^* \tri \max_{ x \in {\cal X}} l(x), \hst \mathcal{U} \tri \Big\{x \in {\cal X}: l(x) = l^* \Big\}$ .
%\ees

\vspace{0.1cm}

\noi Then, the pay-off ${\mathbb L}^{MO}_\alpha({\bf l}, {\bf p})$ can be written as
\begin{align*}
 {\mathbb L}^{MO}_\alpha({\bf l}, {\bf p}) & = \alpha l^*+(1-\alpha) \sum_{x \in {\cal X}} l(x) p(x)  \\
 &=\Big(\alpha+(1-\alpha)\sum_{x \in \mathcal{U}} p(x) \Big) l^*+ \sum_{x \notin \mathcal{U} }(1-\alpha) p(x) l(x)
 \end{align*}
where the set ${\cal U}$ remains to be identified.   Define
\begin{align*}
\sum_{x \in {\mathcal U} }  w_{\alpha}(x)&\tri \Big(\alpha+(1-\alpha)\sum_{x \in \mathcal{U}} p(x) \Big) \\
w_{\alpha}(x)&\tri(1-\alpha) p(x),~ x \notin {\cal U}.
 \end{align*}
 Then the pay-off ${\mathbb L}^{MO}_\alpha({\bf l}, {\bf p})$ can be written as follows:
\bea\label{opt:new}
   {\mathbb L}_{\alpha}^{MO}({\bf l}, {\bf p})=  {\mathbb L}^{MO}({\bf l}, {\bf w}_\alpha) \tri \sum_{x \in {\cal X} }  w_{\alpha}(x) l(x),~\forall \alpha \in [0,1]
\eea
where the weights $w_{\alpha}(x)$ are functions of $\alpha$ and the source probability ${\bf p} \in {\mathbb P}({\cal X})$. It can be easily verified that the new weight vector ${\bf w}_\alpha \tri \{w_\alpha(x): x \in {\cal X}\}$ is a probability distribution since $0 \leq w_{\alpha}(x) \leq 1,~\forall x \in {\cal X}$ and $\sum_{ x \in {\cal X} }w_{\alpha}(x) =1,  \forall \alpha \in [0,1]$.
%Hence, the new pay-off to be minimized is
%\bea
%{\mathbb L}^{MO}({\bf l}, {\bf w}_\alpha) \tri \sum_{x \in {\cal X} }  w_{\alpha}(x) l(x), \quad \forall \alpha \in [0,1].
%\eea
The next lemma describes how the weight vector behaves as a function of the probability vector ${\bf p}$ and $\alpha \in [0,1]$.

% ----------------------------------------------------------------
% Lemma 1
% ----------------------------------------------------------------
\begin{lemma}
\label{lin}
Consider pay-off ${\mathbb L}^{MO}_\alpha({\bf l}, {\bf p})$. Given  any probability distribution ${\mathbb P}({\cal X})$ the following hold. \\
 1. If $p(x)\leq p(y)$, then $w_\alpha(x) \leq w_\alpha(y)$~$\forall x,y \in {\cal X}$, $\alpha \in [0,1]$. Equivalently,  $w_\alpha(x_1) \geq w_\alpha(x_2) \geq \ldots \geq w_\alpha(x_{|{\cal X}|})>0$, for all $\alpha \in [0,1]$.\\
 2. For $y \notin {\cal U}$,  $w_\alpha(y)$ is a monotonically decreasing function of $\alpha \in [0,1]$,  and for $x \in {\cal U}$, $w_\alpha(x)$ is a monotonically increasing  function of $\alpha \in [0,1]$.
\end{lemma}

% ----------------------------------------------------------------
% Proof
% ----------------------------------------------------------------
\begin{proof}
There exist three cases; more specifically, \\
$1)$ $x,y \notin \mathcal{U}$: then $ w_\alpha(x)=(1-\alpha)p(x) \leq (1-\alpha)p(y) = w_\alpha(y)$,~$\forall~\alpha \in [0,1]$; $2)$ $x,y \in \mathcal{U}$:  $w_\alpha(x)=w_\alpha(y) = w_\alpha^{*} \triangleq \min_{ x \in {\cal X}} w_\alpha(x)$; $3)$ $x \in \mathcal{U}$,  $y \notin \mathcal{U}$ (or $x \notin \mathcal{U}$,  $y \in \mathcal{U}$): Consider the case $x \in \mathcal{U}$ and  $y \notin \mathcal{U}$. Then,
\vspace{-0.1cm}
\begin{align}
\frac{\partial w_{\alpha}(y)}{\partial \alpha}&=-p(y)<0,  \label{maxprobab} \\
\frac{\partial w_{\alpha}(x)}{\partial \alpha}&=\frac{1}{|\mathcal{U}|}\frac{\partial w_\alpha^{*}}{\partial \alpha}=\frac{1}{|\mathcal{U}|}\left(1-\sum_{ x \in \mathcal{U}}p(x)\right) >0,  \label{minprobab}
\end{align}
According to \eqref{maxprobab}, \eqref{minprobab}, for $y \notin \mathcal{U}$ the weight $w_{\alpha}(y)$ decreases, and for $x \in \mathcal{U}$ the weight $w_{\alpha}(x)$ increases. Hence, since $w_\alpha(\cdot)$ is a continuous function with respect to $\alpha$, at some $\alpha =\alpha^\prime$, $w_{\alpha^\prime}(x)=w_{\alpha^\prime}(y)=w^*_{\alpha^\prime}$.  Suppose that for some $\alpha=\alpha^\prime+d \alpha$, $d \alpha>0$,  $w_\alpha(x)\neq w_\alpha(y)$.  Then, the largest weight will decrease and the lowest weight will increase as a function of $\alpha \in [0,1]$  according to \eqref{maxprobab} and \eqref{minprobab}, respectively. %2) Follows from 1).
\end{proof}

\vspace{-0.1cm}

%\noi The following corollary is a re-statement of Lemma~\ref{lin}.

%\begin{corollary}\label{cor:1}
%For any probability distribution ${\mathbb P}({\cal X})$ and $\alpha \in [0,1]$, then $w_\alpha(x_1) \geq w_\alpha(x_2) \geq \ldots \geq w_\alpha(x_{|{\cal X}|})>0$.
%\end{corollary}

\begin{remark} (Special Case)
\label{ex1}
Before deriving the general coding algorithm, consider the simplest case when $|\mathcal{U}|=1$, that is $w_\alpha(x_{|{\cal X}|})<w_\alpha(x_{|{\cal X}|-1})$. Then,
\begin{align*}
 {\mathbb L}^{MO}_\alpha({\bf l}, {\bf p}) %& = \alpha l^*+(1-\alpha) \sum_{x \in {\cal X}} l(x) p(x) \nonumber \\
 &=\Big(\alpha+(1-\alpha) p(x_{|{\cal X}|}) \Big) l^*+ \sum_{x \notin \mathcal{U} }(1-\alpha) p(x) l(x) .
%&=\sum_{x  \notin \mathcal{U} }  w_{\alpha}(x)  l(x) + w_{\alpha}(x_{|{\cal X}|})l^*
\end{align*}
%where $w_{\alpha}(x) $ and $w_{\alpha}(x_{|{\cal X}|})$ are the weights for the codeword lengths $l(x)$ and $l^*$, respectively.
\noi In this case, the weights are given by $w_{\alpha}(x) = (1-\alpha)p(x),~x \notin \mathcal{U}$ and $w_{\alpha}(x_{|{\cal X}|}) = \alpha +(1-\alpha)p(x_{|{\cal X}|}) $. This formulation is identical to the minimum expected length problem provided $\alpha \in [0,1]$ is such that  $w_\alpha(x_{|{\cal X}|}) < w_\alpha(x_{|{\cal X}|-1})$. Hence, for any $\alpha \in [0,\alpha_{1})$ defined by
%\vspace{-0.2cm}
\begin{align}\label{alpha1}
% \alpha +(1-\alpha)p(x_{|{\cal X}|})  &\leq (1-\alpha)p(x_{|{\cal X}|-1}) \nonumber \\
\alpha_{1} \triangleq  \frac{p_{|{\cal X}|-1}-p_{|{\cal X}|}}{1+p_{|{\cal X}|-1}-p_{|{\cal X}|}}
\end{align}
%\vspace{-0.1cm}
the codeword lengths are given by $-\log w_\alpha(x), x \in {\cal X}$. For $\alpha \geq \alpha_{1}$ the form of the minimization problem changes, as more weights $w_{\alpha}(x)$ are such that $x \in \mathcal{U}$. The merging rule on the weight vector ${\bf w}_\alpha$ for any $\alpha \in [0,1]$  so that a solution to the coding problem exists for arbitrary cardinality $|{\cal U}|$ and any $\alpha \in [0,1]$ is described next.
\end{remark}

\noi Consider the general case when $|\mathcal{U}| \in \{1, 2, \ldots, |{\cal X}|-1\}$. Define $\alpha_0 \tri 0$ and
\vspace{-0.1cm}
\begin{align*}\label{ak}
\alpha_k &\tri \min\left\{\alpha \in [0,1]: w_\alpha(x_{|{\cal X}|-(k-1)})= w_\alpha(x_{|{\cal X}|-k})\right\}, \\
\Delta \alpha_k &\tri \alpha_{k+1}-\alpha_k, \hst k \in \{1,\ldots, |{\cal X}|-1\}.
\end{align*}
That is, since the weights are ordered as in Lemma \ref{lin}, $\alpha_1$ is the smallest value of $\alpha \in [0,1]$ for which the smallest two weights are equal, $w_\alpha(x_{|{\cal X}|})= w_\alpha(x_{|{\cal X}|-1})$, $\alpha_2$ is the smallest value of $\alpha \in [0,1]$ for which the next smallest two weights are equal, $w_\alpha(x_{|{\cal X}|-1})= w_\alpha(x_{|{\cal X}|-2})$, etc, and $\alpha_{|{\cal X}|-1}$ is the smallest value of $\alpha \in [0,1]$ for which the biggest two weights are equal, $w_\alpha(x_{2})= w_\alpha(x_{1})$. For a given value of $\alpha \in [0,1]$, we define the minimum weight corresponding to a specific symbol in ${\cal X}$  by $w_{\alpha}^*  \tri \min_{ x \in {\cal X}} w_{\alpha}(x)$. \\
\noi Since for $k=0$, $w_{\alpha_0}(x)=w_{0}(x)=p(x), \forall x  \in {\cal X}$, is the set of initial symbol probabilities, let ${\cal U}_0$ denote the singleton set $\{x_{|{\cal X}|} \}$. Specifically,
\vspace{-0.2cm}
\begin{align}
\mathcal{U}_0 \tri \left\{x\in \{x_{|{\cal X}|} \}:  p^{*}\tri \min_{ x\in {\cal X}} p(x)= p(x_{|{\cal X}|})  \right\} .
\end{align}
Similarly, ${\cal U}_1$ is defined as the set of symbols in $\{x_{|{\cal X}|-1}, x_{|{\cal X}|}\}$ whose weight evaluated at $\alpha_1$ is equal to the minimum weight  $w_{\alpha_1}^*$, i.e., 
\bea
\mathcal{U}_1 \tri \Big\{ x \in \{ x_{|{\cal X}|-1}, x_{|{\cal X}|} \}:  w_{\alpha_1}(x)= w_{\alpha_1}^* \Big\}.
\eea
\noi In general, for  a given value of $\alpha_k, k \in \{1,\ldots, |{\cal X}|-1\}$, we define
\begin{align}\label{Uk}
\mathcal{U}_k  \tri \Big\{ x \in \{ x_{|{\cal X}|-k},\ldots, x_{|{\cal X}|} \}:  w_{\alpha_k}(x)= w_{\alpha_k}^* \Big\}.
\end{align}

% ----------------------------------------------------------------
% Lemma 2
% ----------------------------------------------------------------
\begin{lemma}\label{prop1}
Consider pay-off ${\mathbb L}^{MO}_\alpha({\bf l}, {\bf p})$. For any probability distribution ${\bf p} \in {\mathbb P}({\cal X})$ and  $\alpha \in [\alpha_k, \alpha_{k+1}) \subset [0,1]$, $k \in \{0, 1, 2, \ldots, |{\cal X}|-1\}$ then
\begin{align}
w_\alpha(x_{ |{\cal X}|-k}) = w_{\alpha}(x_{|{\cal X}|})=w_{\alpha}^*
\end{align}
and the cardinality of set $\mathcal{U}_k$ is $\left| \mathcal{U}_k \right|=k+1$.
\end{lemma}

% ----------------------------------------------------------------
% Proof
% ----------------------------------------------------------------
\begin{proof}
The validity of the statement  is shown by perfect induction. At $\alpha=\alpha_{1}$,
%\begin{align*}
$w_\alpha(x_{|{\cal X}|})=w_\alpha(x_{|{\cal X}|-1}) \leq w_\alpha(x_{|{\cal X}|-2}) \leq \ldots \leq w_\alpha(x_{1})$.
%\end{align*}
Suppose that, when $\alpha=\alpha_1+d\alpha$, $d\alpha>0$, then $w_\alpha(x_{|{\cal X}|}) \neq w_\alpha(x_{|{\cal X}|-1}) $. Then,
\begin{align}
 {\mathbb L}_\alpha^{MO}({\bf l}, {\bf p}) =\Big(\alpha+(1-\alpha)p(y) \Big) l^*+ \sum_{x \notin \mathcal{U} }(1-\alpha) p(x) l(x) \nonumber
\end{align}
and the weights will be of the form $w_{\alpha}(x) = (1-\alpha)p(x)$ and $w_{\alpha}(y) = \alpha +(1-\alpha)p(y)$ where $y\in\{x_{|{\cal X}|},x_{|{\cal X}|-1} \}$.
Thus,
\begin{align}
\frac{\partial w_{\alpha}(x)}{\partial \alpha}&=-p(x)<0, ~x \notin \mathcal{U} \\
\frac{\partial  w_{\alpha}(y)}{\partial \alpha}&=1-p(y) >0, ~y \in \mathcal{U} \label{minprob}.
\end{align}
Hence, the largest of the two would decrease, while the smallest would increase and therefore they meet again. This contradicts our assumption that  $w_\alpha(x_{|{\cal X}|}) \neq w_\alpha(x_{|{\cal X}|-1}) $ for $\alpha>\alpha_1$. Therefore, $w_\alpha(x_{|{\cal X}|}) = w_\alpha(x_{|{\cal X}|-1}), ~\forall \alpha \in [\alpha_1,1)$.

Secondly, in the case that $\alpha>\alpha_k,, ~k \in \{2,\ldots, {|{\cal X}|}-1\}$, we suppose that the weights 
$w_\alpha(x_{|{\cal X}|})=w_\alpha(x_{|{\cal X}|-1})\ldots= \ldots =w_\alpha(x_{|{\cal X}|-k})=w_{\alpha}^*$. Hence, the pay-off is written as
\begin{align*}%\label{equivalent_multiobj}
 {\mathbb L}_\alpha^{MO}({\bf l}, {\bf p}) =\Big(\alpha+(1-\alpha)\sum_{x \in \mathcal{U} }p(x) \Big) l^*+ \sum_{x \notin \mathcal{U} }(1-\alpha) p(x) l(x)
\end{align*}

\noi Thus,
\begin{align}
\frac{\partial w_{\alpha}(x)}{\partial \alpha}&=-p(x)<0, ~x \notin \mathcal{U} \label{wi_prob} \\
|\mathcal{U} |\frac{\partial w_{\alpha}^*}{\partial \alpha}&=1-\sum_{j=0}^{k}p_{|{\cal X}|-j} >0, ~k \in \{2,\ldots, |{\cal X}|-1\}  .
\end{align}
Finally, in the case that $\alpha>\alpha_{k+1}, ~k \in \{2,\ldots, |{\cal X}|-2\}$, if any of the weights $w_{|{\cal X}|-j}{(\alpha)},~\forall j \in \{ 0, \ldots , k+1 \} $,  changes differently than another, then, either at least one probability will become smaller than others and give a higher codeword length, or it will increase faster than the others and hence according to \eqref{wi_prob} it will decrease to meet the other weights. %as in \eqref{equivalent_multiobj}.
Therefore, the change in this new set of probabilities should be the same, and the cardinality of $\mathcal{U}$ increases by one, i.e., $\mathcal{U}_{k+1}=\left|k+2\right|, ~k\in \{ 2,\ldots |{\cal X}|-2 \}$.
\end{proof}

% ----------------------------------------------------------------
% Main Theorem
% ----------------------------------------------------------------
\noi The main theorem which describes how the weight vector ${\bf w}_\alpha$ changes as a function of $\alpha \in [0,1]$ so that there exist a solution to the coding problem is given in the next theorem.

\begin{theorem}\label{main_theorem}
Consider pay-off ${\mathbb L}^{MO}_\alpha({\bf l}, {\bf p})$. Given a set of probabilities ${\bf p} \in   {\mathbb P}({\cal X})$ and $\alpha \in [\alpha_k,\alpha_{k+1})$, $k\in\{0,1,\ldots,|{\cal X}|-1\}$, the optimal weights ${\bf{w}^{\dagger}_{\alpha}} \tri \{{w}^{\dagger}_{\alpha}(x): x \in {\cal X}  \} \equiv \big({w}^{\dagger}_{\alpha}(x_1), {w}^{\dagger}_{\alpha}(x_2), \ldots, {w}^{\dagger}_{\alpha}(x_{|{\cal X}|})\big)$ are given by
\begin{align}
{\small
w^{\dagger}_{\alpha}(x) =
\begin{cases}
(1-\alpha)p(x),~ x \notin \mathcal{U}_{k}  \\ %=w_{\alpha_k}(x) - (\alpha-\alpha_k)p(x)~, \forall ~x \notin \mathcal{U}_{k}  \\
\displaystyle w_{\alpha_k}^*(x) + (\alpha-\alpha_k) \frac{\sum_{x \notin \mathcal{U}_{k}} p(x) }{|\mathcal{U}_{k}|}, ~x \in \mathcal{U}_{k}
\end{cases}
\label{weights_update}}
\end{align}
where $\mathcal{U}_k$ is given by \eqref{Uk} and
\begin{align}\label{akplus1}
\alpha_{k+1}=\alpha_k+(1-\alpha_k)\frac{(p_{|{\cal X}|-(k+1)}-p_{|{\cal X}|-k})}{\frac{\sum_{x \notin \mathcal{U}_{k}} p(x)}{|\mathcal{U}_{k}|}+p_{|{\cal X}|-(k+1)}} .
\end{align}
\end{theorem}

{\begin{proof}
According to Lemma \ref{prop1}, the lowest probabilities become equal and change together forming a total weight given by
\begin{align*}
\sum_{j=0}^{k}w_{\alpha}(x_{|{\cal X}|-j})&=|\mathcal{U}_{k}| w_{\alpha}^*(x) \\
&= \alpha+(1-\alpha)p_{|{\cal X}|}+\ldots+(1-\alpha)p_{|{\cal X}|-k}.
\end{align*}
\noi Hence,
\begin{align}
|\mathcal{U}_{k}|\frac{w_{\alpha}^*(x)}{\partial \alpha}&=1-\sum_{j=0}^{k}p(x_{|{\cal X}|-j}) \\
\frac{w_\alpha^*(x)}{\partial \alpha}&=\frac{1-\sum_{j=0}^{k}p(x_{|{\cal X}|-j})}{|\mathcal{U}_{k}|}={\frac{\sum_{x \notin \mathcal{U}_{k}} p(x)}{|\mathcal{U}_{k}|}}.
\end{align}
By letting, $\delta_k(\alpha)=\alpha - \alpha_{k}$, then $\forall~ x \in {\mathcal{U}_{k}}$
\begin{align}\label{eq:sol1}
w_\alpha^*(x) =w_{\alpha_k}^*(x) +\delta_k(\alpha){\frac{\sum_{x \notin \mathcal{U}_{k}} p(x)}{|\mathcal{U}_{k}|}},
\end{align}
\noi whereas  $\forall~ x \notin {\mathcal{U}_{k}}$, $w_{\alpha}(x) =(1-\alpha)p(x)$.
      %\begin{align}\label{eq:sol2}
      %w_{\alpha}(x) &=(1-\alpha)p(x)  =\left[1- \left( \alpha_{k}+\delta_k(\alpha) \right)  \right]p(x)  \nonumber \\
      %&=\left(1- \alpha_{k} \right) p(x)-\delta_k(\alpha)p(x)    \nonumber \\
      %&=w_{\alpha_k}(x) -\delta_k(\alpha)p(x).
      %\end{align}
When $\delta_k(\alpha)=\alpha_{k+1}-\alpha_k$, that is $\alpha=\alpha_{k+1}$, then $w_\alpha(x_{|{\cal X}|-(k+1)})=w_{\alpha}^*(x)$ and therefore,
\begin{align*}
&\left(1- \alpha_{k+1} \right) p(x_{|{\cal X}|-(k+1)}) = w_{\alpha_k}^*(x) +\delta_k(\alpha){\frac{\sum_{x \notin \mathcal{U}_{k}} p(x)}{|\mathcal{U}_{k}|}} %\\
     %&=  \left(1- \alpha_{k} \right) p(x_{|{\cal X}|-k}) +(\alpha_{k+1}-\alpha_k){\frac{\sum_{x \notin \mathcal{U}_{k}} p(x)}{|\mathcal{U}_{k}|}}
\end{align*}
and thus, after manipulation $\alpha_{k+1}$ is given by
\begin{align}\label{alpha_k}
\alpha_{k+1}=\alpha_k+(1-\alpha_k)\frac{p(x_{|{\cal X}|-(k+1)})-p(x_{|{\cal X}|-k})}{\frac{\sum_{x \notin \mathcal{U}_{k}} p(x)}{|\mathcal{U}_{k}|}+p_{|{\cal X}|-(k+1)}} .
\end{align}
\end{proof}

% =====================================================
% ROBUST SHANNON CODING
% =====================================================
\section{Optimal Code Lengths}\label{subsec:shannon}

\noi This section presents the optimal real-valued codeword length   vectors ${\bf l} \in {\cal L}({\mathbb R}_+^{|{\cal X}|})$ of the multiobjective pay-offs stated under   Problem~\ref{problem1} and Problem~\ref{problem2}, for any $\alpha \in [0,1]$ and $t \in [0,\infty)$.  \\

%\discuss{Should make $l$ into $l_\alpha$ etc to indicate the parameter dependence????} \TC{I agree, but anyway l depends on w (which depends on $\alpha$). Whatever you prefer.}

% ------------------------------------------------------------
% Theorem
% ------------------------------------------------------------
 \begin{theorem}
 \label{mos}
Consider Problem~\ref{problem1}.   For any probability distribution ${\bf p} \in {\mathbb P}({\cal X})$ and $\alpha \in [0,1]$ the optimal prefix real-valued code ${\bf l} \in {\mathbb {{R}}}_+^{|{\cal X}|}$ minimizing the pay-off ${\mathbb L}^{MO}_{\alpha}({\bf l}, {\bf p})$ is given by
\bes% \label{eq:solutions}
l_{\alpha}^\dagger(x) = \left\{ \begin{array}{lll}
 -\log{\Big((1-\alpha)p(x)\Big)} ~\mbox{for} \: x \notin {\cal U}_k  \\
 -\log{\Big( \frac{\alpha+(1-\alpha)\sum_{x \in \mathcal{U}_k}p(x)}{  |\mathcal{U}_k|} \Big)} ~ \mbox{for} \: x \in {\cal U}_k \end{array} \right.
\ees
where  $\alpha \in [\alpha_k, \alpha_{k+1}) \subset [0,1] ,~\forall k \in \{1, \ldots, |{\cal X}|-1\}$.
 \end{theorem}

 % ------------------------------------------------------------
% Theorem Proof
% ------------------------------------------------------------
\begin{proof}
The pay-off to be minimized is given by \eqref{opt:new}. It can be easily verified that the new weight vector ${\bf w}_\alpha \tri \{w_\alpha(x): x \in {\cal X}\}$ is a probability distribution since $0 \leq w_{\alpha}(x) \leq 1,~\forall x \in {\cal X}$ and $\sum_{ x \in {\cal X} }w_{\alpha}(x) =1,  \forall \alpha \in [0,1]$. Therefore, as in Shannon coding the optimal codeword lengths are given by minus the logarithm of the optimal weights.
\end{proof}

\noi Note that for $\alpha =0$ Theorem~\ref{mos} corresponds to the Shannon solution $l^{sh}(x)=-\log p(x)$, while the solution for $\alpha =1$ is the same as the solution for all $\alpha$ taking values in interval $\alpha \in [\alpha_{|{\cal X}|-1}, 1]$ over which the weight vector ${\bf w}_\alpha$ is identically distributed, and hence $l_\alpha^\dagger(x)|_{\alpha =1}= \frac{1}{|{\cal X}|}$. The behavior of $w_\alpha(x)$ and  $l_\alpha^\dagger(x)$ as a function of $\alpha \in [0,1]$ is described in the next subsection via an illustrative example. The solution of the multiobjective pay-off  ${\mathbb L}{\mathbb R}_{\alpha}({\bf l}+ \log {\bf p}, {\bf p})$ which involves the pointwise redundancy is omitted since it is characterized similarly.

% ------------------------------------------------------------
% Theorem
% ------------------------------------------------------------
 \begin{theorem}
 \label{mos1}
 \noi Consider Problem~\ref{problem2}. For any probability distribution ${\bf p} \in {\mathbb P}({\cal X})$ and $\alpha \in [0,1]$ the optimal prefix real-valued code ${\bf l} \in {\mathbb R}_+^{|{\cal X}|}$ minimizing the pay-off ${\mathbb L}_{t,\alpha}^{MO}({\bf l}, {\bf p})$ is given by
\bea
l_{t,\alpha}^\dagger(x) = -\log \Big( \alpha \nu_t(x) + (1-\alpha) p(x)\Big), \hst \ x \in {\cal X} \label{gsl}
\eea
where $\{\nu_{t,\alpha}(x): x \in {\cal X}\}$ is defined via the tilted probability distribution
\bea
\nu_{t,\alpha}(x) \tri \frac{D^{ t \: l_{t,\alpha}^\dagger(x)} p(x)}{ \sum_{ x \in {\cal X}} p(x) D^{ t \: l_{t,\alpha}^\dagger(x)} }, \hst   x \in {\cal X}  \label{gs11}
\eea
 \end{theorem}

% ------------------------------------------------------------
% Theorem Proof
% ------------------------------------------------------------
\begin{proof}
% Since ${\mathbb L}_{t,\alpha}^{MO}({\bf l}, {\bf p})$ is a convex function and the constraint set is convex, this is  a convex optimization problem.
By invoking the  Karush-Kuhn-Tucker necessary and sufficient conditions  of  optimality one obtains the following set of equations  describing the optimal codeword lengths.
\bea
D^{-l_{t,\alpha}^\dagger(x)}=\alpha \nu_{t,\alpha}(x) + (1-\alpha )p(x), x \in {\cal X} \label{gs}
\eea
which gives (\ref{gsl}).
\end{proof}

\noi Note that the solution stated under Theorem~\ref{mos1} corresponds, for $\alpha=0$ to the Shannon code, which minimizes the average codeword length pay-off, while for $\alpha=1$ (after manipulations) it is given by
\begin{align}
 l_{t,\alpha=1}^{\dagger}(x) = - \frac{1}{1+t} \log p(x)  + \log \Big( \sum_{x \in {\cal X}} p(x)^\frac{1}{1+t}   \Big), \hso x \in {\cal X}  \nonumber
\end{align}
which is precisely the solution of a variant of the Shannon code,  minimizing the average of an exponential function of the codeword length pay-off \cite{1981:humblet_generalization,2008b:Baer}.
\noi The solution of the multiobjective Payoff ${\mathbb L}_{t,\alpha}^{MO}({\bf l}+\log_D {\bf p}, {\bf p})$ corresponding to pointwise redundancy is obtained similarly as in Theorem~\ref{mos1}. The optimal codeword lengths are given by
\bea
l_{t,\alpha}^\dagger(x) = -\log \Big( \alpha \mu_{t,\alpha}(x) + (1-\alpha) p(x)\Big), \hst \ x \in {\cal X} \label{gsl1}
\eea
where $\{\mu_{t,\alpha}(x): x \in {\cal X}\}$ is defined via the tilted probability distribution
\bea
\mu_{t,\alpha}(x)= \frac{D^{ t \: l_{t,\alpha}^\dagger(x)} p^{t+1}(x)}{ \sum_{ x \in {\cal X}} p^{t+1}(x) D^{ t \: l_{t,\alpha}^\dagger(x)} } \hst  x \in {\cal X}. \label{gs}
\eea
The only difference between the optimal codeword lengths of pay-off ${\mathbb L}_{t,\alpha}^{MO}({\bf l}+\log_D {\bf p}, {\bf p})$ with respect to the pay-off ${\mathbb L}_{t,\alpha}^{MO}({\bf l}, {\bf p})$  is the term $p^{t+1}(x)$ appearing in the tilted distribution. When $\alpha=1$ (\ref{gsl1}) is precisely a a variant of the Shannon code,  minimizing the average of an exponential function of the redundancy of the codeword length pay-off \cite{2008a:Baer,2006a:Baer}.

\begin{remark} 1. {\it The Limiting Case as $t \rar \infty$}: The minimization of the multiobjective pay-off ${\mathbb L}_{\alpha}^{MO}({\bf l}, {\bf p})$ obtained in Theorem~\ref{mos} is indeed obtained from the minimization of the two parameter multiobjective pay-off ${\mathbb L}_{t,\alpha}^{MO}({\bf l}, {\bf p})$ in the limit, as $t \rar \infty$.
%, as follows. Write $\sum_{ x \in {\cal X}}p(x) D^{t \: l_{t,\alpha}^\dagger(x)}= \sum_{ x \notin {\cal U}} p(x) D^{t \: l_{t,\alpha}^\dagger(x)} + \sum_{ x \in {\cal U}} p(x) D^{t \: l_{t,\alpha}^\dagger(x_u)}$.Then, in the limit as, $t \rar \infty$,  (\ref{gs}) is equivalent to the following expressions.
%\begin{align}
%D^{-l_\alpha^\dagger(x)} & = (1-\alpha )p(x), \hst   x \notin {\cal U}_k \label{gss1} \\
%D^{-l_\alpha^\dagger(x)} &=\alpha \frac{ p(x)}{ \sum_{ x \in {\cal U}_k} p(x) } + (1-\alpha )p(x), \hst   x \in {\cal U}_k \label{gss2}
%\end{align}
%Since $p(x)=p(y), \forall x, y \in {\cal U}_k$ then (\ref{gss1}) and (\ref{gss2}) correspond to \TC{\eqref{weights_update}}.
In addition, $\lim_{t \rar \infty} {\mathbb L}_{t,\alpha}^{MO}({\bf l}, {\bf p})={\mathbb L}_{\alpha}^{MO}({\bf l}, {\bf p})$, $\forall {\bf l}$ and hence at ${\bf l} ={\bf l}^\dagger$. The point to be made here is that the solution of Problem~\ref{problem1} can be deduced from the solution of Problem~\ref{problem2} in the limit as $t \rar \infty$ provided the merging rule on how the solution changes with $\alpha \in [0,1]$ is employed. \\
2. {\it Coding Theorems}: Although, coding theorems for Problem~\ref{problem1} and Problem~\ref{problem2} are not presented (due to space limitation), these can be easily obtained either from the closed form solutions or by following  \cite{1965:campbell_coding}.
\end{remark}

% =====================================================
% FAST ALGORITHM FOR FINDING OPTIMAL WEIGHTS
% =====================================================
\subsection{An Algorithm for Computing the Optimal Weights}\label{subsec:algorithm}
\label{fa}
\noi For any probability distribution  ${\bf p} \in {\mathbb P}({\cal X})$ and $\alpha \in [0,1]$ an algorithm is presented to compute the optimal weight vector ${\bf w}_\alpha$ for any $\alpha \in [0,1]$.

\begin{figure}[H]
\centering
\includegraphics[width=\columnwidth]{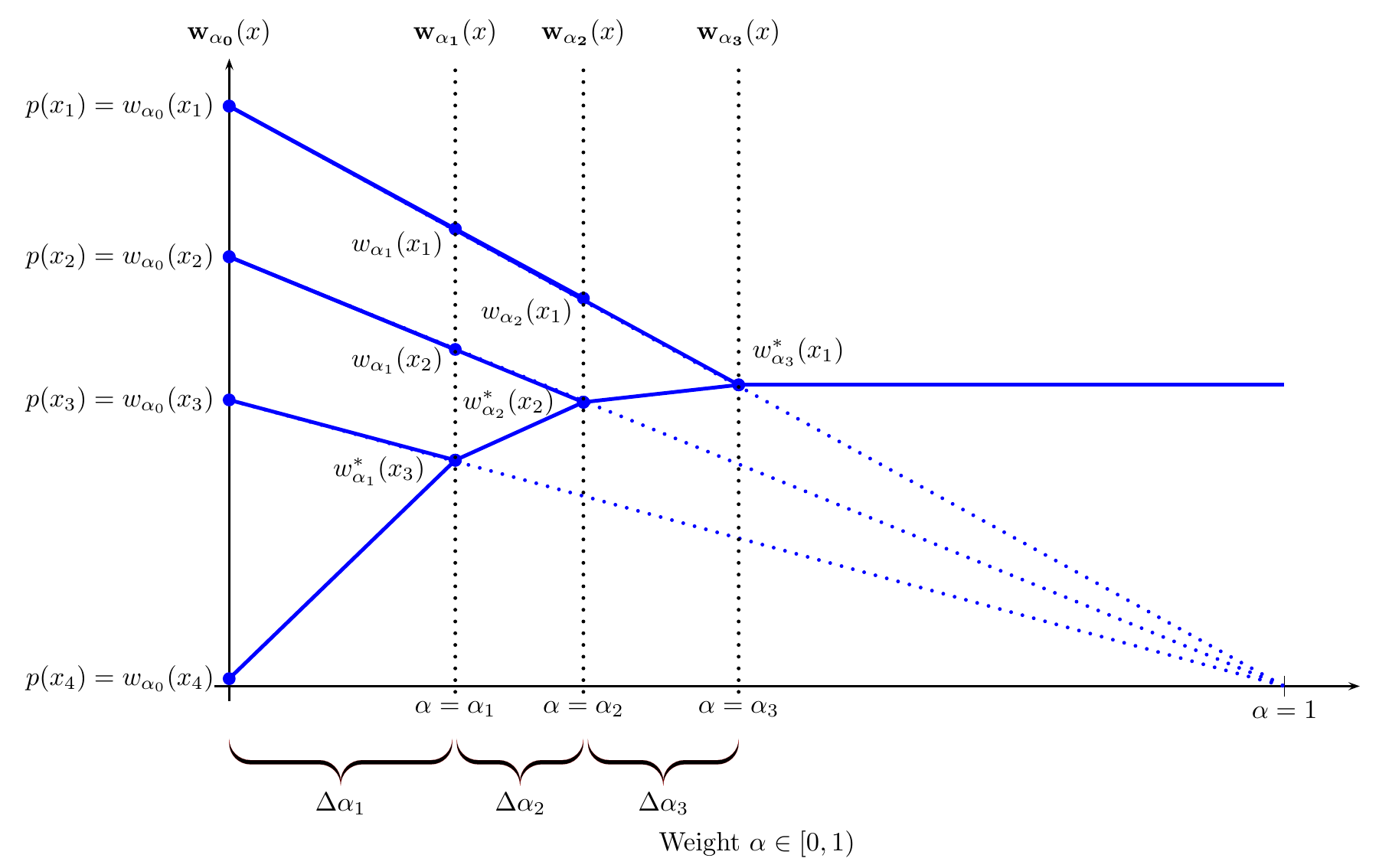}
\caption{A schematic representation of the weights for different values of $\alpha$.}\label{probs1}
\end{figure}

\noi It is shown in Section~\ref{ow} (see also Figure~\ref{probs1}) that the weight vector ${\bf w}_\alpha$ changes piecewise linearly as a function of  $\alpha \in [0,1]$. Therefore, to calculate the weights $w_{\hat{\alpha}}(x)$ for a specific value of $\hat{\alpha} \in [0,1]$, one is only required to determine  the values of $\alpha$ at the intersections by using \eqref{akplus1}, up to the intersection (see Fig.\ref{probs1}) that gives a value greater than $\hat{\alpha}$ or up to the last intersection (if all the intersections give a smaller value of $\alpha$). Thus, one can easily find the weights at $\hat{\alpha}$ by using \eqref{weights_update}. %The algorithm is depicted under Algorithm \ref{charalambous_algorithm} below.

\begin{algorithm}
\caption{\small Algorithm for Computing the Weight Vector ${\bf w}_\alpha$ }
\label{charalambous_algorithm}
\begin{algorithmic}
\STATE $\,$\\
\STATE \textbf{initialize}
\STATE $\quad\, \mathbf{p}=\left(p(x_1), p(x_2), \ldots, p(x_|{\cal X}|)\right)^T$, $\alpha = \hat{\alpha}$
%\STATE $\quad\, \alpha = \hat{\alpha}$
\STATE $\quad\, k=0$,  $\alpha_0 = 0$
%\STATE $\displaystyle \quad\,  \alpha_{1}=\frac{p_{n-1}-p_{n}}{1+p_{n-1}-p_{n}}$
%\STATE $\,$
\WHILE{$\displaystyle \hat{\alpha}>\alpha_{k} $}
%\STATE $\,$
%\STATE {Calculate $w_{i}^{*}(\alpha_k) $:}
%\STATE $\quad\, \displaystyle w_{i}^{*}(\alpha_k)=\left( 1- a_{k} \right)p_{n-k}$
%\STATE $\,$
\STATE {Calculate $\alpha_{k+1}$:}
\STATE {$\quad\, \displaystyle \alpha_{k+1}= \alpha_{k}+(1-\alpha_k)\frac{p(x_{|{\cal X}|-(k+1)})-p(x_{|{\cal X}|-k})}{\frac{\sum_{x \notin \mathcal{U}_{k}} p(x)}{k+1}+p(x_{|{\cal X}|-(k+1)})}  $}
%\STATE $\,$
\STATE {$k \leftarrow k + 1$}
\ENDWHILE
%\STATE $\,$
\STATE {$k \leftarrow k - 1$}
%\STATE{$ \displaystyle \delta_k(\alpha)=\alpha-\alpha_{k}$}
%\STATE $\,$
\STATE {Calculate $\mathbf{w}^{\dagger}_{\hat{\alpha}}$:}
\FOR{$v = 1$ to $|{\cal X}|-(k+1)$}
\STATE $w^{\dagger}_{\hat{\alpha}}(x_{v})=(1-\hat{\alpha})p(x_{v})$
\STATE $v \leftarrow v + 1$
\ENDFOR
%\STATE $\,$
\STATE {Calculate $w^{*}_{\hat{\alpha}}(x) $:}
\STATE $\quad\, \displaystyle w^{*}(\hat{\alpha}) =\left( 1- a_{k} \right)p(x_{|{\cal X}|-k})+ (\hat{\alpha}-\alpha_k) \frac{\sum_{x \notin \mathcal{U}_{k}} p(x) }{k+1} $
%\STATE $\,$
\FOR{$v = |{\cal X}|-k$ to $|{\cal X}|$}
\STATE $\displaystyle w^{\dagger}(x_{v})=w^{*}_{\hat{\alpha}}(x)$
\STATE $v \leftarrow v+ 1$
\ENDFOR
%\STATE $\,$
\RETURN $\mathbf{w}^{\dagger}_{\hat{\alpha}}$.
\end{algorithmic}
\end{algorithm}

% =====================================================
%
%
% EXAMPLES
%
%
% =====================================================
\subsection{Illustrative Example}\label{sec:examples}
%This section describes specific examples.

\noi Consider binary codewords and a source with   $|{\cal X}|=4$ and probability distribution $\displaystyle \mathbf{p}=\left(\begin{array}{cccc}
   \frac{8}{15} &  \frac{4}{15} &  \frac{2}{15} &  \frac{1}{15}
\end{array}\right)$. Using the  algorithm one can find the optimal weight vector $\mathbf{w}^{\dagger}$ for different values of $\alpha \in [0,1]$ for which pay-off (\ref{b30}) of Problem~\ref{problem1} is minimized. Compute  $\alpha_1$ via \eqref{akplus1}, $\alpha_1=1/16$.
%\begin{align*}
%\alpha_1=\frac{\frac{2}{15}-\frac{1}{15}}{1-\frac{1}{15}+\frac{2}{15}}=\frac{1}{16}.
%\end{align*}
For  $\alpha=\alpha_1 =1/16$ the optimal weights are 
\begin{align*}
&w_3^{\dagger}(\alpha)=w_4^{\dagger}(\alpha)=(1-\alpha)p_3=\frac{1}{8} \\
&w_2^{\dagger}(\alpha)=(1-\alpha)p_2=\frac{1}{4} \\
&w_1^{\dagger}(\alpha)=(1-\alpha)p_1=\frac{1}{2}
\end{align*}
 In this case, the resulting codeword lengths correspond to the optimal Huffman code.
The weights for all $\alpha \in [0,1]$ can be calculated iteratively by calculating $\alpha_k$ for all $k\in \{ 0, 1, 2, 3\}$ and noting that the weights vary linearly with $\alpha$ (Figure \ref{probsEx1}).

\begin{figure}[H]
\centering
\includegraphics[width=\columnwidth]{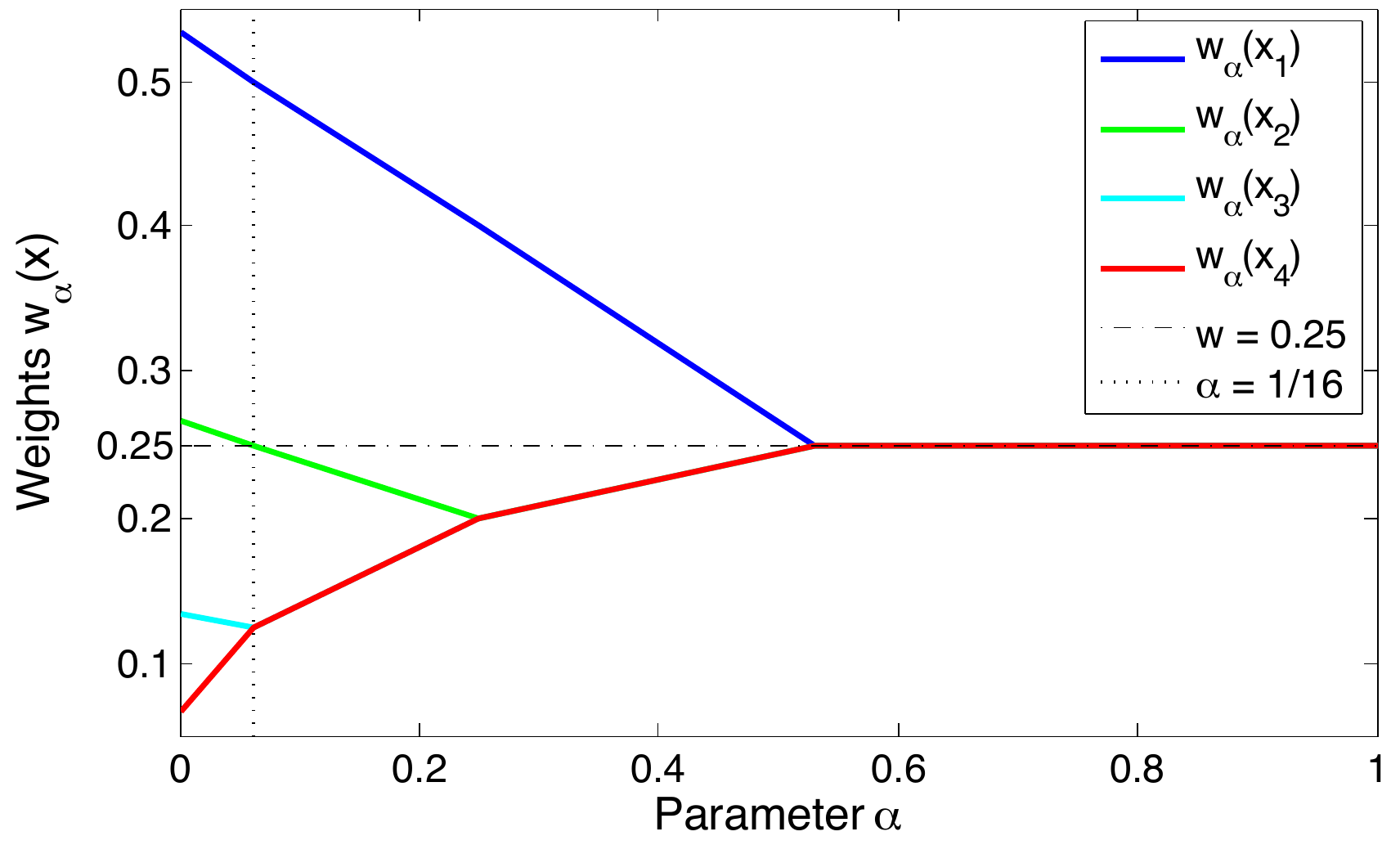}
\caption{A schematic representation of the weights for different values of $\alpha$ when $p=(\frac{8}{15},\frac{4}{15},\frac{2}{15},\frac{1}{15})$.}\label{probsEx1}
\end{figure}

% =====================================================
%
%
% CONCLUSION
%
%
% =====================================================
\section{Conclusion and Future Directions}\label{sec:conclusions}
Two lossless coding problems with multiobjective pay-offs are investigated and the idealized real-valued codeword length solutions are presented. Relations to problems discussed in the literature are obtained. Based on the insight gained in this paper, Huffman like algorithms which solve this problem are part of ongoing research.
% Although, Huffman like algorithms which solve these problems are not discussed it is believed that such algorithms can be found as well based on the insight gained in this paper.

% =====================================================
%
%
% ACKNOWLEDGMENT
%
%
% =====================================================
%\section*{Acknowledgment}


% -------------------------------------------------------------
% Bibliography
% -------------------------------------------------------------
\bibliographystyle{IEEEtran}
\bibliography{bibliografia}

\begin{thebibliography}{10}
\providecommand{\url}[1]{#1}
\csname url@samestyle\endcsname
\providecommand{\newblock}{\relax}
\providecommand{\bibinfo}[2]{#2}
\providecommand{\BIBentrySTDinterwordspacing}{\spaceskip=0pt\relax}
\providecommand{\BIBentryALTinterwordstretchfactor}{4}
\providecommand{\BIBentryALTinterwordspacing}{\spaceskip=\fontdimen2\font plus
\BIBentryALTinterwordstretchfactor\fontdimen3\font minus
  \fontdimen4\font\relax}
\providecommand{\BIBforeignlanguage}[2]{{%
\expandafter\ifx\csname l@#1\endcsname\relax
\typeout{** WARNING: IEEEtran.bst: No hyphenation pattern has been}%
\typeout{** loaded for the language `#1'. Using the pattern for}%
\typeout{** the default language instead.}%
\else
\language=\csname l@#1\endcsname
\fi
#2}}
\providecommand{\BIBdecl}{\relax}
\BIBdecl

\bibitem{2006:Cover}
T.~M. Cover and J.~A. Thomas, \emph{Elements of Information Theory},
  2nd~ed.\hskip 1em plus 0.5em minus 0.4em\relax Wiley-Interscience, 2006.

\bibitem{2004:DrmotaSzpankowski}
M.~Drmota and W.~Szpankowski, ``Precise minimax redundancy and regret,''
  \emph{IEEE Transactions of Information Theory}, vol.~50, pp. 2686--2707,
  2004.

\bibitem{2008a:Baer}
M.~Baer, ``Tight bounds on minimum maximum pointwise redundancy,'' in
  \emph{{IEEE} International Symposium on Information Theory}, july 2008, pp.
  1944 --1948.

\bibitem{1965:campbell_coding}
L.~Campbell, ``A coding theorem and {R}$\acute{e}$nyi's entropy,''
  \emph{Information and Control}, vol.~8, no.~4, pp. 423--429, Aug. 1965.

\bibitem{1981:humblet_generalization}
P.~Humblet, ``Generalization of huffman coding to minimize the probability of
  buffer overflow,'' \emph{{IEEE} Transactions on Information Theory}, vol.~27,
  no.~2, pp. 230--232, 1981.

\bibitem{2008b:Baer}
M.~Baer, ``Optimal {P}refix {C}odes for {I}nfinite {A}lphabets {W}ith
  {N}onlinear {C}osts,'' \emph{{IEEE} Transactions on Information Theory},
  vol.~54, no.~3, pp. 1273 --1286, march 2008.

\bibitem{2006a:Baer}
------, ``A general framework for codes involving redundancy minimization,''
  \emph{IEEE Trans. of Information Theory}, vol.~52, pp. 344--349, 2006.

\bibitem{1973:davisson}
L.~Davisson, ``Universal noiseless coding,'' \emph{Information Theory, IEEE
  Transactions on}, vol.~19, no.~6, pp. 783--795, Nov 1973.

\bibitem{1980:davisson_Leon-Garcia}
L.~Davisson and A.~Leon-Garcia, ``A source matching approach to finding minimax
  codes,'' \emph{Information Theory, IEEE Transactions on}, vol.~26, no.~2, pp.
  166--174, Mar 1980.

\bibitem{2007:FarzadCharalambous}
C.~Charalambous and F.~Rezaei, ``Stochastic uncertain systems subject to
  relative entropy constraints: Induced norms and monotonicity properties of
  minimax games,'' \emph{Automatic Control, {IEEE} Transactions on}, vol.~52,
  no.~4, pp. 647--663, April 2007.

\bibitem{2009:Gawrychowski_Gagie}
P.~Gawrychowski and T.~Gagie, ``Minimax trees in linear time with
  applications,'' in \emph{Combinatorial Algorithms}, J.~Fiala,
  J.~Kratochv\'{\i}l, and M.~Miller, Eds.\hskip 1em plus 0.5em minus
  0.4em\relax Berlin, Heidelberg: Springer-Verlag, 2009, pp. 278--288.

\bibitem{2005:rezaei_bambos}
F.~Rezaei and C.~Charalambous, ``Robust coding for uncertain sources: a minimax
  approach,'' in \emph{Proceedings. International Symposium on Information
  Theory, {ISIT}}, 2005, pp. 1539--1543.

\end{thebibliography}

%\begin{thebibliography}{1}


%\end{thebibliography}

\end{document}